\documentclass[conference,a4paper]{IEEEtran}

\usepackage[utf8]{inputenc} 
\usepackage[T1]{fontenc}
\usepackage{url}              
\usepackage{cite}             

\usepackage[cmex10]{amsmath}  
\usepackage{mleftright}       
\mleftright                   

\usepackage{graphicx}         
\usepackage{booktabs}         

\usepackage{amsmath, amssymb, tikz, enumitem, tabularx, bm, amsthm, mathabx, verbatim, import, subfiles, subfigure, mathtools, tikz-cd}

\newtheorem{proposition}{Proposition}
\newtheorem{lemma}[proposition]{Lemma}

\newtheorem{theorem}[proposition]{Theorem}

\newtheorem{corollary}[proposition]{Corollary}

\theoremstyle{definition}
\newtheorem{definition}[proposition]{Definition}
\newtheorem{example}[proposition]{Example}
\newtheorem{remark}[proposition]{Remark}

\newcommand{\N}{\ensuremath{\mathbb{N}}}
 \newcommand{\R}{\ensuremath{\mathbb{R}}}
 \newcommand{\C}{\ensuremath{\mathbb{C}}}

\newcommand{\phid}{\ifmmode \Phi \text{ID} \else $\Phi$ID \fi}

\newcommand{\Bint}{\ensuremath{B^\circ}}
\newcommand{\Lint}{\ensuremath{L^\circ}}

\newcommand{\LintD}{\ensuremath{\frac{\partial \Lint}{\partial x}}}

\DeclareMathOperator{\content}{\mathcal{C}}
\DeclareMathOperator{\Dis}{Dis}

\DeclareMathOperator{\Hom}{Hom}
\DeclareMathOperator{\id}{id}

\usetikzlibrary{decorations.markings}

\begin{document}

\title{A Logarithmic Decomposition for Information}

%
%
%
 \author{
   \IEEEauthorblockN{Keenan J. A. Down\IEEEauthorrefmark{1}\IEEEauthorrefmark{2} and 
                     Pedro A. M. Mediano\IEEEauthorrefmark{3}}
   \IEEEauthorblockA{\IEEEauthorrefmark{1}%
                     Queen Mary, University of London, k.j.a.down@qmul.ac.uk}
   \IEEEauthorblockA{\IEEEauthorrefmark{2}%
                     University of Cambridge, kjad2@cam.ac.uk}
   \IEEEauthorblockA{\IEEEauthorrefmark{3}%
                     Imperial College London, p.mediano@imperial.ac.uk}
 }

\maketitle

\begin{abstract}
The Shannon entropy of a random variable $X$ has much behaviour analogous to a signed measure. Previous work has concretized this connection by defining a signed measure $\mu$ on an \textit{abstract information space} $\tilde{X}$, which is taken to represent the information that $X$ contains. This construction is sufficient to derive many measure-theoretical counterparts to information quantities such as the mutual information $I(X; Y) = \mu(\tilde{X} \cap \tilde{Y})$, the joint entropy $H(X,Y) = \mu(\tilde{X} \cup \tilde{Y})$, and the conditional entropy $H(X|Y) = \mu(\tilde{X}\, \setminus \, \tilde{Y})$. We demonstrate that there exists a much finer decomposition with intuitive properties which we call the \textit{logarithmic decomposition (LD)}. We show that this signed measure space has the useful property that its \textit{logarithmic atoms} are easily characterised with negative or positive entropy, while also being coherent with Yeung's $I$-measure \cite{yeung1991new}. We present the usability of our approach by re-examining the G\'acs-K\"orner common information from this new geometric perspective and characterising it in terms of our logarithmic atoms. We then highlight that our geometric refinement can account for an entire class of information quantities, which we call \textit{logarithmically decomposable} quantities. 
\end{abstract}

\section{Introduction}
\IEEEPARstart{F}{or} all first-order information-theoretical quantities derived from the classical Shannon entropy on a series of random variables $X_1,\ldots, X_r$, Yeung demonstrated that there exists a representative set in a corresponding $\sigma$-algebra $\mathcal{F}$ and, moreover, that for any set in the $\sigma$-algebra there is a sensible measure of information \cite{yeung1991new}. This correspondence, built on earlier work by Hu Kuo Ting in \cite{ting1962amount}, offers a firm foundation for the measure-theoretical perspective of entropy.

The $\sigma$-algebra $\mathcal{F}$ of Yeung is coarse in that it is generated by the unions, intersections, and complements of \textbf{abstract information spaces} $\tilde{X}_1,\ldots, \tilde{X}_r$. This symbolic connection, while mechanically useable and consistent, offers little insight into the constituent elements of these information spaces. Our geometric perspective on a refinement provides both a quantitative and qualitative foundation for this measure.

Decomposing these information spaces is of great interest across multiple domains. What kind information is transmitted across a network of neurons, and with what qualitative structure? How can we disentangle the complex interplay between confounding variables, such as gender and pay, or race and arrest rate? It is known that the mutual information cannot generally be encoded \cite{gacs1973common}, but can we develop alternative language to explain these interactions? In this work we describe these abstract information spaces in greater detail than, to the best of our knowledge, has previously been seen. Given a series of random variables $X_1, \ldots, X_r$ we present a theoretically maximal refinement of the corresponding $\sigma$-algebra. We will construct a measure $\Lint$ on this abstract information space, and this measure shall represent the informational content of its subsets. In doing so we also decompose the $\sigma$-algebra of Yeung \cite{yeung1991new} into \textit{logarithmic} atoms, whose contribution to the entropy is particularly easy to characterise, in a process we call \textbf{logarithmic decomposition}.

We will make the utility of our new vocabulary clear by also expressing the common information of G\'acs and K\"orner \cite{gacs1973common} in terms of our logarithmic decomposition. We will see that both mutual information and common information reside in a class of information quantities we call \textbf{logarithmically decomposable} quantities, which we believe to contain many standard measures.

\section{Refinement of Abstract Information Spaces}
\label{InformationSpaces}
Let $X_1,\ldots, X_r$ be discrete random variables on a corresponding finite sample space $\Omega$ with the natural $\sigma$-algebra $\mathcal{F}$ generated by all possible combinations of outcomes on each variable. Using the probability space $(\Omega, \mathcal{F}, P)$, we will define a corresponding space for information.

\begin{definition}
Let $(\Omega, \mathcal{F}, P)$ be a probability space as above. Then define the \textbf{complex} of $\Omega$ as the simplicial complex on all outcomes $\omega \in \Omega$:

\begin{equation}
\Delta(\Omega) = \bigcup_{k = 1}^N \Omega_k \cong \mathcal{P}(\Omega) \setminus \{\varnothing\}.
\end{equation}

where $\Omega_k$ is the set of subsets $S \subseteq \Omega$ with $|S| = k$ and $N = |\Omega|$. For a collection of $n$ outcomes $\omega_1,\ldots, \omega_n$, we label the corresponding simplex for $n\geq 2$ as $\Bint(\omega_1,\ldots, \omega_n) \in \Omega^n$, which, viewing $\Delta(\Omega)$ geometrically, corresponds to a face, volume, or edge without its boundaries. For additional consistency, we let $B(\varnothing) = \Bint(\varnothing) = \emptyset$. Note that we write $\Bint$ to signify that the boundaries of the simplex are not included.

We occasionally write $\Delta\Omega$ instead of $\Delta(\Omega)$ to simplify the notation.
\end{definition}

\begin{example}
Consider a space of outcomes $\Omega = \{1, 2, 3, 4\}$. The complex consists of the following elements

\begin{align}
\begin{split}
\Delta(\Omega) = \{& B(1), B(2), B(3), B(4), \\
& \Bint(1,2), \Bint(1,3), \Bint(1,4),\\
& \Bint(2,3), \Bint(2,4), \Bint(3,4), \\
& \Bint(1,2,3), \Bint(1,2, 4), \Bint(1, 3, 4), \Bint(2, 3, 4),  \\
& \Bint(1, 2, 3, 4)\}
\end{split}
\end{align}

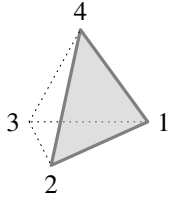
\begin{figure}[h]
\centering
\begin{tikzpicture}[line join = round, line cap = round]

\coordinate [label=above:4] (4) at (0,{sqrt(2)},0);
\coordinate [label=left:3] (3) at ({-.5*sqrt(3)},0,-.5);
\coordinate [label=below:2] (2) at (0,0,1);
\coordinate [label=right:1] (1) at ({.4*sqrt(3)},0,-.5);

\begin{scope}[decoration={markings,mark=at position 0.5 with {\arrow{to}}}]
\draw[dotted] (1)--(3);
\draw[very thick, gray, fill=lightgray,fill opacity=.5] (2)--(1)--(4)--cycle;
\draw[dotted] (3)--(2);
\draw[dotted] (3)--(4);
\end{scope}

\end{tikzpicture}

\caption{The highlighted triangle along with its boundary corresponds to the subset $B(1,2,4) = \{B(1), B(2), B(4), \Bint(1, 2), \Bint(1, 4), \Bint(2, 4), \Bint(1, 2, 4)\}$.}
\end{figure}

\end{example}

\begin{remark}
The points of the simplex (of the form $B(\omega)$) are relatively inconsequential and we will make use of them only when advantageous, but otherwise we will not mention them further. We will see in section \ref{MeasureSpace} that the points will be associated with zero entropy in our constructed measure space.
\end{remark}

\section{Construction of a signed measure}
\label{MeasureSpace}
Having endowed $\Delta(\Omega)$ with a geometric interpretation, we would like to equip it now with a signed measure. Doing this completes the construction of the \textbf{abstract information space} $\tilde{X}$ of a random variable $X$, which has appeared previously in the form of a set-variable correspondence \cite{yeung1991new, ting1962amount}. We will use the entropy loss to assign an entropy to each atom in our decomposition \cite{baez2011characterization}.

We will define here two measures of entropy loss: The \textbf{total entropy loss} $L$, which will represent the total entropy lost when merging a select group of regions in a partition; and the \textbf{interior loss} $\Lint$, which we will later use to measure the elements of our geometric space $\Delta(\Omega)$. We will see that $L$ measures entire simplices with their boundaries, and $\Lint$ will push the construction further to measure the interiors of these simplices. The purpose of this distinction is that the $L$ measure alone offers enough behaviour to deduce the measure of Campbell \cite{campbell1965entropy}, but becomes insufficient to study non-orthogonal partitions of $\Omega$. The interior loss measure $\Lint$ will refine this measure to resolve this orthogonality issue.

From the perspective of entropy loss, a variable will \textit{lose} entropy when boundaries between events are deleted \cite{baez2011characterization}. To see this, let $X$ be a random variable corresponding to a partition $\bm{Q}_X = \{Q_1, \ldots, Q_t\}$ of the outcome space $\Omega$ where $\mathbb{P}(Q_i) = \sum_{\omega \in Q_i} \mathbb{P}(\omega)$. If we create a new random variable $X'$ by merging two of the partitions $Q_1$ and $Q_2$ so that $\bm{Q}_{X'} = \{Q_1 \cup Q_2, Q_3, \ldots, Q_t\}$ becomes the new partition, then the new variable $X'$ will have a reduced entropy. Clearly, removing all boundaries will correspond to an entropy loss equivalent to the total entropy of the variable $X$.

\begin{definition}
\label{DEFINITIONTotalLoss}
Let $X$ be a random variable with corresponding partition $\bm{Q}_X = \{Q_1, \ldots, Q_t\}$, and let $X'$ be the random variable with corresponding partition 

\begin{equation}
\bm{Q}_{X'} =  \left\{\bigcup_{\alpha \in A} Q_\alpha \right\} \cup \left\{ Q_\beta : \beta \notin A \right\},
\end{equation}

where $A$ is a subset of events which we intend to merge. What we recover is $\bm{Q}_X$ with all parts indexed in $A$ merged together. Then we define the \textbf{total loss}

\begin{equation}
\label{EqnEntropyChange}
L(A) = H(X) - H(X').
\end{equation}

We may simplify the situation by writing $L(p_1,\ldots, p_n)$ instead of  $L(P_1,\ldots, P_n)$ or writing $L(S)$ for $S = \{p_1,\ldots, p_n\}$ to signify that $L$ can also be viewed as a function on $\R^n$. Expanding the above expression we find

\begin{align}
\begin{split}
L(p_1,\ldots, p_n) =&\, H(X) - H(X') \\
= &\, p_1 \log \left(\frac{1}{p_1}\right) + \cdots + p_n \log \left(\frac{1}{p_n}\right) \\
 &\, - (p_1 + \cdots + p_n) \log \left( \frac{1}{p_1 + \cdots + p_n} \right) \\
= &\, \log \left[ \frac{(p_1 + \cdots + p_n)^{(p_1 + \cdots + p_n)}}{p_1^{p_1} \ldots p_n^{p_n}} \right] .
\end{split}
\end{align}

\end{definition}

We note that for any partition $P_1,\ldots, P_n$ of $\Omega$ we must have that $L(P_1,\ldots, P_n)\geq 0$. Moreover, using equation \eqref{EqnEntropyChange}, it is immediately clear that for a random variable $X$ with outcomes of associated probabilities $p_1,\ldots, p_n$ with $\sum p_i = 1$, we have

\begin{equation}
\label{EqnEntropyEqualToTotalLoss}
H(X) = L(p_1,\ldots, p_n)
\end{equation}

Trivially we also see that $L(p) = 0$ for any $p\in [0, 1]$.

We now extend the definition of the total loss using M\"obius inversion to produce the \textbf{interior loss}, $\Lint$. For maximum strength, we will now treat $\Omega$ as a partition of singletons $\omega_i \in \Omega$, as we will see this is sufficient to describe all other possible partitions.

\begin{definition}
We will define the \textbf{interior loss} function $\Lint(\omega_1,\ldots, \omega_n)$ recursively. For $n=1$ let $\Lint(\omega) = 0$. For $n\geq 1$ we then recursively define $\Lint$ by

\begin{equation}
\Lint(\omega_1,\ldots, \omega_n) = L(\omega_1,\ldots, \omega_n) - \sum_{\substack{S \subset \{\omega_1,\ldots, \omega_n\} \\ |S| \leq n-1}} \Lint(S).
\end{equation}

Again, as with the total loss, we will often abuse this notation and write $\Lint(p_1,\ldots, p_n)$ where the probabilities reflect individual outcomes or regions in the partition.
\end{definition}

Geometrically, we will see that the interior loss will measure entropies in interior regions of the complex $\Delta(\Omega)$.

\begin{remark}
\label{InclusionExclusionIdentities}
The total loss can be expressed as a sum of the interior losses by virtue of their construction:
\begin{equation}
L(\omega_1,\ldots, \omega_n) = \sum_{S \subseteq \{\omega_1,\ldots, \omega_n\}} \Lint(S),
\end{equation}

and hence the interior loss function can also be expressed in terms of the loss function by virtue of the inclusion-exclusion principle:

\begin{equation}
\label{InclusionExclusion}
\Lint(\omega_1,\ldots, \omega_n) = \sum_{S \subseteq \{\omega_1,\ldots, \omega_n\}} (-1)^{n-|S|} L(S).
\end{equation}

The interior loss corresponds to the M\"obius inversion of the total loss on the partially ordered set defined by containment of simplices.
\end{remark}

As an aside, and while meaningless within probability theory, we note that the functions $L$ and $\Lint$ can both be extended to domains where the \textit{probabilities} $p_i$ are greater than one, and many of the following results hold for any $p_i \in \R^+$.

We now show that $\Lint$ can be used as a measure of entropy. We will later demonstrate that $\Lint$ provides a refinement of the $I$-measure of Yeung \cite{yeung1991new}.

\begin{theorem}
\label{LintSignedMeasure}
Let $\Omega$ be a finite set of outcomes and let $\Sigma$ be the $\sigma$-algebra generated by all of the elements $b\in \Delta(\Omega)$. For $S\subseteq \Delta(\Omega)$ define $\Lint(S) = \sum_{b\in S} \Lint(b)$. Then $(\Delta(\Omega), \Sigma, \Lint)$ is a finite signed measure space.
\end{theorem}

\begin{proof}
Setting $\Lint(\varnothing) = 0$, and using the definition of $\Lint(S)$ we see that $\Lint$ is at least countably additive across disjoint sets in $\Sigma$. Hence $(\Delta \Omega, \Sigma, \Lint)$ is a signed measure space.
\end{proof}

We will see in section \ref{SUBSECTION_mutual_information} that this measure is a refinement of the the $I$-measure of Yeung \cite{yeung1991new}. For now, we shall state a series of further results which are useful for calculations with $\Lint$. 

\begin{lemma}[Interior loss identity]
\label{LEMMA_InteriorLossIdentity}
Let $T = \{p_1,\ldots, p_k\}$ for some collection of probabilities. For notational clarity we will write

\begin{equation}
\sigma(T) = \sigma(p_1,\ldots, p_k) = (p_1 + \cdots + p_k)^{(p_1 + \cdots + p_k)}.
\end{equation}

Further still we shall write

\begin{equation}
A_k = \prod_{\substack{S \subseteq\{p_1,\ldots, p_n\} \\ |S| = k}} \sigma(S).
\end{equation}

Then we have that

\begin{equation}
\label{ClearerExpression}
\Lint(p_1,\ldots, p_n) = \sum_{k=1}^n (-1)^{n-k}\log (A_k)
\end{equation}

\end{lemma}

This lemma demonstrates that our atoms are measured by alternating sums of logarithms, justifying the name \textbf{logarithmic decomposition}. Our next lemma allows the confident inclusion of $0$ in our domain.

\begin{lemma}[Interior loss at 0]
\label{InteriorLossAt0}
For $p_1,\ldots, p_n, x \in \R^+$ where $n\geq 0$, we have

\begin{equation}
\lim_{x\to 0} \Lint(p_1,\ldots, p_n, x) = 0
\end{equation}
\end{lemma}

\begin{lemma}
\label{LintAtInfinity}
Let $p_1,\ldots, p_{n-1}, x \in \R^+$ and let $x$ vary. Then

\begin{equation}
\lim_{x \to \infty} | \Lint(p_1,\ldots, p_{n-1}, x)| = |\Lint(p_1,\ldots, p_{n-1})|
\end{equation}

\end{lemma}

The result of this lemma will be useful for a corollary which follows from the next result. The following theorem demonstrates the useful property that logarithmic atoms have an intrinsic sign.

\begin{theorem}
\label{AlternatingDerivatives}
Let $p_2,\ldots, p_{n} \in \R^+$ be a sequence of nonzero arguments for $n\geq 2$ and $m\geq 0$. Then

\begin{equation}
(-1)^{m+n} \frac{\partial^m \Lint}{\partial x^m}(x, p_2,\ldots, p_n) \geq 0.
\end{equation}

\end{theorem}

Setting $m=0$ we immediately see that the sign of logarithmic atoms alternates solely on the number of outcomes they contain; a property which standard co-informations do not have.

\begin{corollary}[Interior magnitude can only decrease]
\label{InteriorMagnitude}
Let $p_1,\ldots, p_{n-1}, \tau \in \R^+\cup\{0\}$ for $n\geq 3$. Then
\begin{equation}
|\Lint(p_1,\ldots, p_{n-1}, \tau)| < |\Lint(p_1,\ldots, p_{n-1})|
\end{equation}
\end{corollary}

This result is quite powerful in that it works for $p_1,\ldots, p_{n-1}, \tau \in [0,\infty)$. For our information-theoretical purposes, we will naturally require that $p_i\in [0,1]$, so the measure of successively higher-order volumes in $\Delta(\Omega)$ will strictly decrease, with the slowest descent for $p_1 = \cdots = p_{n}$.

\section{Quantities of information}

Having constructed the measurable space $\Delta(\Omega)$ we will now demonstrate its utility in characterising various variable-level information quantities.

Firstly we will show how mutual information and co-information can be reinterpreted using the logarithmic decomposition, and we show that the information measure $\Lint$ is consistent with the prevailing measure of Yeung \cite{yeung1991new}.

\subsection{Mutual and Co-information}
\label{SUBSECTION_mutual_information}
Suppose we have two variables $X$ and $Y$ defined on the same outcome space $\Omega$. This outcome space can be taken to be the meet of the two partitions corresponding to $X$ and $Y$ if necessary. 

The degree to which the two variables interact can be quantified in terms of their entropies via their mutual information, $I(X;Y)$, which is naturally derived as $I(X;Y) := H(X) + H(Y) - H(X,Y)$. This expression can be derived from multiple perspectives, including homologically, as shown by Baudot and Bennequin \cite{baudot2015homological}. Alternatively, it is given by the Kullback-Leibler divergence between the joint distribution and the product of the marginal distributions $D_{KL}(P(X,Y) || P(X)P(Y))$ \cite{cover1991elements}, that is, it captures the degree to which the joint distribution diverges from independence.

One possible generalisation of the mutual information for multiple variables (although others exist, for example the \textit{total correlation} \cite{watanabe1960information} or the \textit{dual total correlation} \cite{te1978nonnegative}) is the \textbf{interaction information} or \textbf{co-information} \cite{mcgill1954multivariate, bell2003co}. This can be defined recursively using

\begin{equation}
I(X_1;\ldots ;X_r) = I(X_1;\ldots; X_{r-1}) - I(X_1;\ldots; X_{r-1} | X_r).
\end{equation}

The co-information can be taken as representing the information that is common to several variables. The next definition will give us the connection between a random variable and its logarithmic decomposition. Doing this will enable us to discuss the co-information in terms of logarithmic atoms.

\begin{definition}
\label{DEFINITION_content}
Given a random variable $X$, we define the \textbf{content} $\content(X)$ inside of $\Delta(\Omega)$ to be the set of all boundaries crossed by $X$. That is, if $X$ corresponds to a partition $P_1,\ldots, P_n$, then

\begin{multline}
\mathcal{C}(X) = \{ \text{$\Bint(S): S \subseteq \Omega, \exists\,  \omega_i,\omega_j \in S$} \\
\text{with $\omega_i \in P_k$, $\omega_j\in P_l$ such that $k\neq l$ }\}.
\end{multline}

Intuitively, this means that at least two of the outcomes in $\Bint(\omega_1,\ldots, \omega_n)$ correspond to distinct events in $X$. We will in general make use of $\mathcal{C}$ to represent the logarithmic decomposition functor from random variables to their corresponding sets in $\Delta\Omega$.
\end{definition}

\begin{example}
To demonstrate our refinement, we consider the space $\Omega = \{1,2,3,4\}$. Let the partitions be given by $X =\{\{1, 3\}, \{2, 4\}\}$ and $Y = \{\{1, 2\}, \{3, 4\}\}$, as in figure \ref{FIG_two_random_vars}.

Taking the intersection of the contents $\mathcal{C}(X) \cap \mathcal{C}(Y)$ gives the logarithmic content mutual to both variables. These logarithmic atoms are given in figure \ref{XYLogarithmic}.

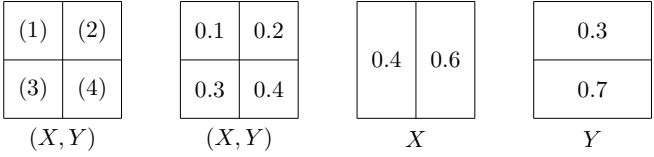
\begin{figure}[h]
\centering
    \scalebox{0.88}{
    \begin{tikzpicture}[scale=0.88]
    \draw (0,0) -- (2, 0);
    \draw (0,0) -- (0, -2);
    \draw (0, -2) -- (2, -2);
    \draw (2, -2) -- (2, 0);
    \draw (0, -1) -- (2, -1);
    \draw (1, 0) -- (1, -2);
    \node at (0.5, -0.5) {$(1)$};
    \node at (1.5, -0.5) {$(2)$};
    \node at (0.5, -1.5) {$(3)$};
    \node at (1.5, -1.5) {$(4)$};
    \node at (1, -2.35) {$(X, Y)$};

    \draw (3,0) -- (5, 0);
    \draw (3,0) -- (3, -2);
    \draw (3, -2) -- (5, -2);
    \draw (5, -2) -- (5, 0);
    \draw (3, -1) -- (5, -1);
    \draw (4, 0) -- (4, -2);
    \node at (3.5, -0.5) {$0.1$};
    \node at (4.5, -0.5) {$0.2$};
    \node at (3.5, -1.5) {$0.3$};
    \node at (4.5, -1.5) {$0.4$};
    \node at (4, -2.35) {$(X, Y)$};

    \draw (6,0) -- (8, 0);
    \draw (6,0) -- (6, -2);
    \draw (6, -2) -- (8, -2);
    \draw (8, -2) -- (8, 0);
    \draw (7, 0) -- (7, -2);
    \node at (6.5, -1) {$0.4$};
    \node at (7.5, -1) {$0.6$};
    \node at (7, -2.35) {$X$};

    \draw (9,0) -- (11, 0);
    \draw (9,0) -- (9, -2);
    \draw (9, -2) -- (11, -2);
    \draw (11, -2) -- (11, 0);
    \draw (9, -1) -- (11, -1);
    \node at (10, -0.5) {$0.3$};
    \node at (10, -1.5) {$0.7$};
    \node at (10, -2.35) {$Y$};

    \end{tikzpicture}
    }
\caption{Two random variables on the set $\Omega = \{1, 2, 3, 4\}$ with some illustrative probabilities.}
\label{FIG_two_random_vars}
\end{figure}

\begin{figure}[!t]
\centering
        \begin{tikzpicture}
        \newcommand{\spaceblubber}{0.35}
        \newcommand{\outerspace}{2.7}
        \newcommand{\innerspace}{0.9}
        \newcommand{\circlesize}{5.5cm}
        \newcommand{\circledisplacement}{0.9}
        
        \node [draw, circle, minimum size = \circlesize, label={135:$X$}] (X) at (-\circledisplacement,0){};
        
        \node [draw, circle, minimum size = \circlesize, label={45:$Y$}] (Y) at (\circledisplacement,0){};
        
        \node at (-\outerspace, \spaceblubber) {$\Bint(1, 2)$};
        \node at (-\outerspace, -\spaceblubber) {$\Bint(3, 4)$};
        \node at (\outerspace, \spaceblubber) {$\Bint(1, 3)$};
        \node at (\outerspace, -\spaceblubber) {$\Bint(2, 4)$};
        
        \node at (0, 5*\spaceblubber) {$\Bint(1, 4)$};
        \node at (0, 3*\spaceblubber) {$\Bint(2, 3)$};
        \node at (-\innerspace, 1*\spaceblubber) {$\Bint(1, 2, 3)$};
        \node at (\innerspace, 1*\spaceblubber) {$\Bint(1, 2, 4)$};
        \node at (-\innerspace, -1*\spaceblubber) {$\Bint(1, 3, 4)$};
        \node at (\innerspace, -1*\spaceblubber) {$\Bint(2, 3, 4)$};
        \node at (0, -3*\spaceblubber) {$\Bint(1, 2, 3, 4)$};
        \end{tikzpicture}
        
        \caption{Logarithmic atoms of $\Omega = \{1,2, 3, 4\}$ in the space $\Delta(\Omega)$.}
        \label{XYLogarithmic}
\end{figure}
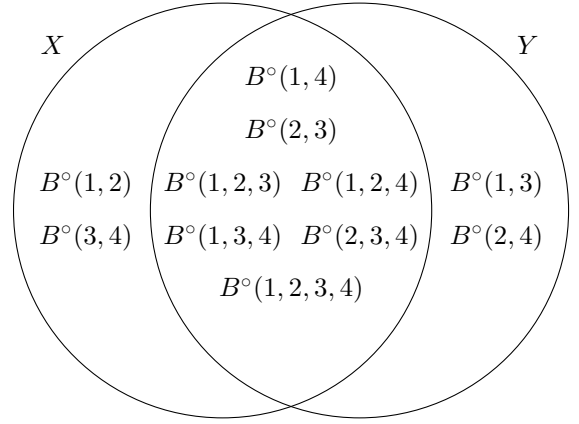

\end{example}

The next theorem is the main result of this paper, demonstrating that Logarithmic Decomposition is at least consistent with the standard atomic decomposition of Yeung \cite{yeung1991new}. 

\begin{theorem}
\label{THM_yeung_correspondence}
Let $R$ be a region on an $I$-diagram of variables $X_1,\ldots, X_r$ with Yeung's $I$-measure. Then
\begin{equation}
I(R) = \sum_{B \in \content(R)} \Lint(B).
\end{equation}
That is, the interior loss measure is consistent with Yeung's $I$-measure.
\end{theorem}

We note also that all inclusion-exclusion expressions such as that in equation \eqref{Eqn_standard_mutual_information} can be naturally extended to the formal sum used in the proof. Provided all of the coefficients in the formal sum are either 0 or 1, we can assign the informational quantity with a \textit{content}, i.e. a subset of $\Delta(\Omega)$. For example, the O-information of Rosas et al. \cite{rosas2019quantifying} satisfies this property and hence has an associated content.

\subsection{Common Information}

An intrinsic problem in the study of random variables is that interactions between variables often cannot be encoded with a third variable \cite{gacs1973common}. The G\'acs-K\"orner formulation of this \textbf{common information} has, for instance, been shown to have little relation to the mutual information in most scenarios.

We have seen in section \ref{SUBSECTION_mutual_information} that mutual information can be completely described by the intersection of variable contents in logarithmic decomposition. We will demonstrate that the logarithmic decomposition can also describe the common information of G\'acs and K\"orner.

To do this, we will demonstrate that the common information shared between a finite collection of variables $X_1,\ldots, X_r$ corresponds to a subset of $\content(X_1)\cap \cdots \cap \content(X_r)$.

\begin{definition}[G\'acs-K\"orner Common Information]
The G\'acs-K\"orner common information on a finite set of random variables $X_1,\ldots, X_r$ \cite{gacs1973common} is given by

\begin{multline}
C_{GK}(X_1;\ldots; X_r) = \max_{Z} H(Z) \\ \text{such that $f_1(X_1) = \cdots = f_r(X_r) = Z$ for some $f_i$.}
\end{multline}
\end{definition}

This common information encodes interactions between variables which can be extracted and represented by another variable \cite{yu2016generalized}. The common information is known to be usually far less than the mutual information \cite{gacs1973common}. We now demonstrate that common information can be represented as a subset in $\Delta(\Omega)$.

\begin{theorem}
\label{THM_gacs_korner}
The G\'acs-K\"orner common information of a finite set of variables $X_i$ corresponds to the maximal subset $C$ of $\bigcap_{i}\content(X_i)$ such that there exists some random variable $Z$ with $\content(Z) = C$.
\end{theorem}

It will perhaps be useful to discuss variables which are contained in arbitrary subsets of $\Delta(\Omega)$. For this purpose, we give the following definitions.

\begin{definition}
Given a subset $R \subseteq \Delta(\Omega)$, we say that $R$ is \textbf{discernible} if it corresponds to the content of any random variable $Z$.

Moreover, given any subset $S\subseteq \Delta(\Omega)$, let $\Dis(S) \subseteq S$ be the largest discernible subset of $S$. We will call this the \textbf{maximally discernible subset of $S$}.
\end{definition}

Note that $\Dis(S)$ is well defined as the trivial random variable is always discernible in $S$, and we also have uniqueness. To see this, note that if two such non-isomorphic variables $Z_1$ and $Z_2$ were to exist, then $\content(Z_1Z_2) \subseteq S$ would be a larger subset, contradicting their maximality.

\begin{remark}
As seen in theorem \ref{THM_gacs_korner}, $\Lint[\Dis(\bigcap_i \content(X_i))] = C_{GK}(X_1;\ldots; X_r)$, the G\'acs-K\"orner common information.
\end{remark}

For an example illustrating this result geometrically, see figure \ref{FIGURE_triangles_intersect}.

\newcommand{\grayopac}{80}
\newcommand{\varlabdist}{0.2}
\newcommand{\chrr}{0.4}

\begin{figure}[!t]
\centering
    \begin{tikzpicture}[line join = round, line cap = round]
    \begin{scope}[shift={(-10,1)}, scale=1]
    \scalebox{1.0}{
        \draw[very thick] (6,0 + \chrr) -- (8,0 + \chrr) -- (8,-2 + \chrr) -- (6, -2 + \chrr) -- cycle;
        \draw[very thick] (7,0 + \chrr) -- (7,-1 + \chrr) -- (6, -1 + \chrr);
        \draw[very thick] (7, -1 + \chrr) -- (7, -2 + \chrr);
        \node at (6.5, -0.5 + \chrr) {1};
        \node at (7.5, -0.5 + \chrr) {2};
        \node at (6.5, -1.5 + \chrr) {3};
        \node at (7.5, -1.5 + \chrr) {4};
        \node at (7, -2.55 - \varlabdist + \chrr) {$X$};

        \begin{scope}[shift={(0, -0.5)}]
        \draw[very thick] (6,-3) -- (8,-3) -- (8,-5) -- (6, -5) -- cycle;
        \draw[very thick] (7,-3) -- (7, -4);
        \draw[very thick] (6, -4) -- (8, -4);
        \node at (6.5, -3.5) {1};
        \node at (7.5, -3.5) {2};
        \node at (6.5, -4.5) {3};
        \node at (7.5, -4.5) {4};
        \node at (7, -5.55 - \varlabdist) {$Y$};
        \end{scope}
        }
    \end{scope}
    \begin{scope}
        \coordinate [label=above:4] (4) at (0,{sqrt(2)},0);
        \coordinate [label=left:3] (3) at ({-.5*sqrt(3)},0,-.5);
        \coordinate [label=below:2] (2) at (0,0,1);
        \coordinate [label=right:1] (1) at ({.4*sqrt(3)},0,-.5);
    
        \begin{scope}
        \draw[dashed, gray!\grayopac] (4)--(2);
        \draw[very thick, black] (1)--(3);
        \draw[very thick, black] (2)--(1)--(4);
        \draw[very thick, black] (3)--(2)--cycle;
        \draw[very thick, black] (3)--(4)--cycle;
        \node at (0, -1.3) {a) $\content(X)$};
        \end{scope}
    \end{scope}
    \begin{scope}[shift={(0,-4)}]
        \coordinate [label=above:4] (4) at (0,{sqrt(2)},0);
        \coordinate [label=left:3] (3) at ({-.5*sqrt(3)},0,-.5);
        \coordinate [label=below:2] (2) at (0,0,1);
        \coordinate [label=right:1] (1) at ({.4*sqrt(3)},0,-.5);
    
        \begin{scope}
        \draw[dashed, gray!\grayopac] (3)--(4);
        \draw[very thick, black] (1)--(3);
        \draw[very thick, black] (2)--(1)--(4)--cycle;
        \draw[very thick, black] (3)--(2)--cycle;
        \node at (0, -1.3) {b) $\content(Y)$};
        \end{scope}
    \end{scope}
    \begin{scope}[shift={(3,0)}]
    
        \coordinate [label=above:4] (4) at (0,{sqrt(2)},0);
        \coordinate [label=left:3] (3) at ({-.5*sqrt(3)},0,-.5);
        \coordinate [label=below:2] (2) at (0,0,1);
        \coordinate [label=right:1] (1) at ({.4*sqrt(3)},0,-.5);
    
        \begin{scope}
        \draw[dashed, gray!\grayopac] (3)--(4)--(2);
        \draw[very thick, black] (1)--(3);
        \draw[very thick, black] (2)--(1)--(4);
        \draw[very thick, black] (3)--(2)--cycle;
        \node at (0, -1.3) {c) $\content(X)\cap \content(Y)$};
        \end{scope}
    \end{scope}
    \begin{scope}[shift={(3, -4)}]
        \coordinate [label=above:4] (4) at (0,{sqrt(2)},0);
        \coordinate [label=left:3] (3) at ({-.5*sqrt(3)},0,-.5);
        \coordinate [label=below:2] (2) at (0,0,1);
        \coordinate [label=right:1] (1) at ({.4*sqrt(3)},0,-.5);
    
        \begin{scope}
        \draw[dashed, gray!\grayopac] (2)--(3)--(4)--cycle;
        \draw[very thick, black] (2)--(1)--(4);
        \draw[very thick, black] (1)--(3);
        \node at (0, -1.3) {d) $\Dis(\content(X)\cap \content(Y))$};
        \end{scope}
    \end{scope}
\end{tikzpicture}
\caption{The 1-dimensional atomic contents of $X$ and $Y$ and their (c) intersection and (d) maximally discernible subset.}
\label{FIGURE_triangles_intersect}
\end{figure}
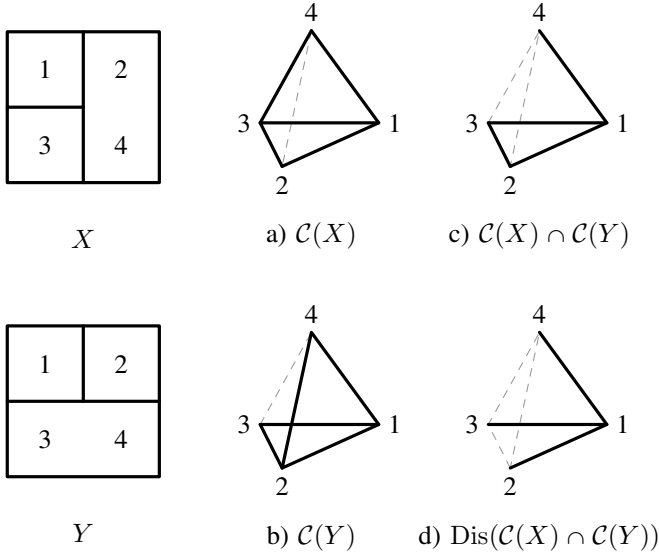

\subsection{Logarithmically decomposable quantities}

Having expressed the mutual information and the common information in terms of subsets of $\Delta(\Omega)$, we are inclined to expect that other kinds of information quantities can also be captured by the logarithmic decomposition. To that end, we give two definitions.

\begin{definition}
Given a collection of random variables $\{X_\alpha: \alpha \in A\}$ for some index set $A$ on a common outcome space $\Omega$, we let $\mathcal{A} = \mathcal{P}(A)$ be the powerset of $A$, and we define a \textbf{variable quantity} to be any map $f: \mathcal{A} \to \mathbb{C}$.
\end{definition}

This notion of variable quantity encapsulates any kind of computation on  any set of random variables and hence certainly contains information quantities. Note that we might allow $f$ to only be defined for finite or countable domains if necessary.

\begin{definition}
Let $f: \mathcal{A} \to \mathbb{C}$ be a variable quantity defined on an outcome space $\Omega$. We say that $f$ is \textbf{logarithmically decomposable} if there exists a mapping $f^*$, defined whenever $f$ is defined, sending sets of variable contents in $\content \mathcal{A}$ to corresponding subsets of $\Delta \Omega$ such that the diagram commutes.

\begin{equation}
\begin{tikzcd}[column sep = 5, row sep = 25]
\mathcal{A} \arrow[rr, "f"] \arrow[d, "\content", style=swap] &\,& \mathbb{C} \\
\mathcal{CA} \arrow[rr, "f*"] &\,& \mathcal{P}(\Delta\Omega) \arrow[u, "\Lint", style=swap]
\end{tikzcd}
\end{equation}

We denote the set of logarithmically decomposable variable quantities by $\mathcal{L}(A)$ or just $\mathcal{L}$.
\end{definition}

This definition captures the idea that the function $f$ can be evaluated by computing a logarithmic decomposition in some form and then applying the measure $\Lint$.

We have seen that all mutual and co-informations are logarithmically decomposable. Moreover, in the previous subsection we also saw that the G\'acs-K\"orner common information is also logarithmically decomposable.

\section{Conclusion}
In this paper we demonstrated that there exists a signed measure space for information which is finer than the $I$-measure of Yeung \cite{yeung1991new}. To our knowledge, this signed measure space is the finest such space treated in the literature, although we expect it can be rederived from many perspectives.

To complete this refinement we utilised a M\"obius inversion on the lattice of all possible outcomes $\Delta(\Omega)$. To demonstrate that this approach indeed reflects a more fundamental description of entropy, we showed that the \textit{sign} of the logarithmic atoms $\Bint(\omega_1,\ldots, \omega_n)$ is intrinsic to their structure. Our proof was analytic, but we expect a proof using convexity is possible. We demonstrated that this approach can capture both mutual information and common information, and we noted that there is a larger class of information quantities which can be logarithmically decomposed. We expect many standard information quantities are also logarithmically decomposable, but for space we have not treated others.

This work opens a new combinatorical framework for use in studying logarithmically decomposable information quantities. While the work presented here takes a geometric perspective, entropy on an inter-variable scale has been shown to have much homological behaviour \cite{baudot2015homological, vigneaux2017information}, and we hope that our tools might be incorporated into the more general topological study of information in the future.

Lastly, we hope that these atoms might be applied to the problem of \textbf{partial information decomposition} \cite{williams2010nonnegative} and extensions such as $\Phi$ID \cite{mediano2021towards}, as the result of theorem \ref{AlternatingDerivatives} demonstrates a potential new avenue for the characterisation of \textit{synergistic} interactions in complex systems.

\paragraph*{Acknowledgements}The authors would like to thank Dan Bor, Fernando Rosas, Robin Ince and Juho Äijälä for interesting discussions on this work and its future directions.

\IEEEtriggeratref{9}

\clearpage
\appendices
\section*{Proofs for results}

\subsection*{Proof of lemma \ref{LEMMA_InteriorLossIdentity}}
\begin{proof}
To simplify we shall also write $f_k = \begin{pmatrix} n-1 \\ k-1 \end{pmatrix}$. This is the number of subsets $S\subseteq \{p_1,\ldots, p_n\}$ of size $k$ which contain a given $p_i$. As we ask for subsets which already contain $p_i$, this is equivalent to asking how many subsets there are of size $k-1$ in $\{p_1,\ldots, p_n\} \setminus \{p_i\}$.

Taking equation \eqref{InclusionExclusion} and using the definition of the total loss function we have

\begin{multline}
    \Lint(p_1, \ldots, p_n) \\
= \log \left[  \frac{A_n}{\sigma(p_1) \ldots \sigma(p_n)} \cdot \frac{\sigma(p_1)^{f_{n-1}} \ldots \sigma(p_n)^{f_{n-1}}} {A_{n-1}} \cdots \right. \\ 
\left.\cdots \left( \frac{A_1}{\sigma(p_1)^{f_1} \ldots \sigma(p_n)^{f_1}}\right) ^{(-1)^{n-1}} \right] \\
= \sum_{k = 1}^n {(-1)^{n-k}} \log \left[ \frac{A_k}{\sigma(p_1)^{f_k}\ldots \sigma(p_n)^{f_k}} \right] \\
\end{multline}

Notice that $f_1 = 1$ so that the final term in this sequence with $k=1$ is equal to $\log(1) = 0$. Counting the powers of $\sigma(p_i)$ shows that in the final expression the power of $\sigma(p_i)$ will be $f_n - f_{n-1} + f_{n-2} + \cdots \pm f_2$ (as the $k=1$ term is cancelled by $A_1$). It is a standard result that

\begin{align}
\begin{split}
\label{BinomialStandardResult}
\sum_{k=1}^n (-1)^{(n-k)}f_k &= 0 \quad \text{and hence} \\
\sum_{k=2}^n (-1)^{(n-k)}f_k &= (-1)^{n}
\end{split}
\end{align}

Hence in the final expression the power of $\sigma(p_i)$ is $(-1)^n$. Rewriting $\sigma(p_1)\cdots \sigma(p_n)= A_1$ gives us the result of equation \eqref{ClearerExpression}.
\end{proof}

\subsection*{Proof of lemma \ref{InteriorLossAt0}}
\begin{proof}
We are augmenting $p_1,\ldots, p_n$ with the additional argument $x$, where we will allow $x$ to vary. Let us now write

\begin{equation}
B_k = \prod_{\substack{S \subseteq\{p_1,\ldots, p_n, x\} \\ x \in S \\ |S| = k}} \sigma(S).
\end{equation}

Then equation \eqref{ClearerExpression} becomes

\begin{multline}
\label{BkAndAk}
\Lint(p_1,\ldots, p_n, x) = \sum_{k=1}^{n+1} (-1)^{n+1-k} \log(B_k(x))  \\
+  \sum_{k=1}^{n} (-1)^{n+1-k} \log(A_k)
\end{multline}

Here we take $A_k$ to be a product of all terms not containing the argument $x$ as per lemma \ref{InteriorLossAt0}. We notice that the sign of all terms $A_k$ have now flipped, but are otherwise identical. We want to show that as $x\to 0$ that these two sums will cancel. Recall that $B_k(x)$ is a product of terms of the form $\sigma(p_1,\ldots, p_n, x) = (p_1 + \cdots + p_n+ x)^{(p_1 + \ldots + p_n + x)}$ for subsets of size $k$. We see that

\begin{equation}
\lim_{x\to 0} \sigma(p_1,\ldots, p_n, x) = \sigma(p_1,\ldots, p_n)
\end{equation}

By the product and quotient rules for limits, we hence also have that

\begin{equation}
\lim_{x\to 0} B_k = A_{k-1}
\end{equation}

Inserting this into equation \eqref{BkAndAk} we see that both sides immediately cancel to give zero as $x\to 0$.
\end{proof}

\subsection*{Proof of lemma \ref{LintAtInfinity}}

\begin{proof}
Using the expression of lemma \ref{LEMMA_InteriorLossIdentity} and the notation for $B_k(x)$ from lemma \ref{InteriorLossAt0} we can write

\begin{multline}
\Lint(p_1,\ldots, p_{n-1}, x) = \sum_{k=1}^n (-1)^{n-k}\log (B_k(x)) + \\
\sum_{k=1}^{n-1} (-1)^{n-k}\log (A_k)
\end{multline}

Where we have omitted the term in $A_n$ because any subset of $\{p_1,\ldots, p_{n-1}, x\}$ of size $n$ is certain to contain $x$. We immediately see that the second expression is equal to $-\Lint(p_1,\ldots, p_{n-1})$. It therefore suffices to show that the first expression in the $B_k(x)$ tends to 0 as $x\to \infty$

Writing the logarithm of $B_k(x)$ as a single fraction, we know by the standard binomial result in equation \eqref{BinomialStandardResult} that the number of factors on the top and the bottom of the fraction containing $x$ is equal. Let the number of factors be $m$. Then, expanding the expression in $B_k(x)$, we see it is dominated on the top and the bottom by an $x^m$ term. This term will dominate as $x\to \infty$, so that the fraction tends to 1 and the logarithm in $x$ will tend to 0, leaving us with

\begin{equation}
\lim_{x \to \infty} \Lint(p_1,\ldots, p_{n-1}, x) = -\Lint(p_1,\ldots, p_{n-1}),
\end{equation}

giving the result immediately.
\end{proof}

\subsection*{Proof of theorem \ref{AlternatingDerivatives}}

\begin{proof}
We will prove this by induction on $n$. To start, we demonstrate that the derivative of $\Lint$ has some useful properties. Using standard results and utilising the notation of lemma \ref{LEMMA_InteriorLossIdentity}, we have that

\begin{multline}
\frac{\partial}{\partial x} \, \sigma(x, p_2,\ldots, p_{k}) = \sigma(x, p_2,\ldots,p_{k}) \\ \cdot\left[ \log(x + p_2+\cdots+p_{k}) + 1 \right]
\end{multline}

We restate the identity in equation \eqref{BkAndAk} for $n-1$ fixed probabilities:

\begin{multline*}
\Lint(p_1,\ldots, p_{n-1}, x) = \sum_{k=1}^{n} (-1)^{n-k} \log(B_k(x))  \\
+ \sum_{k=1}^{n-1} (-1)^{n-k} \log(A_k)
\end{multline*}

The second sum does not depend on $x$. Differentiating with respect to $x$ we obtain

\begin{align}
\begin{split}
&\frac{\partial \Lint}{\partial x}(p_1,\ldots, p_{n-1}, x) \\
=& \sum_{\substack{S \subseteq \{p_1,\ldots, p_{n-1}, x\} \\ x\in S}} (-1)^{n-|S|}\frac{\partial}{\partial x} \, \log(\sigma(S))\\
=& \sum_{\substack{S \subseteq \{p_1,\ldots, p_{n-1}, x\} \\ x\in S}} (-1)^{n-|S|} \frac{\sigma'(S)}{\sigma(S)} \\
=& \sum_{\substack{S \subseteq \{p_1,\ldots, p_{n-1}, x\} \\ x\in S}} (-1)^{n-|S|} \left[ \log\left( \sum_{s \in S} s \right) + 1 \right] \\
\end{split}
\end{align}

The total number of subsets $S\subseteq \{p_1,\ldots, p_{n-1}\}$ of size $k$ is $\begin{pmatrix} n-1 \\ k \end{pmatrix}$, so by the standard result in equation \eqref{BinomialStandardResult} the $+1$ terms will cancel leaving only an alternating sum of logarithms.

To simplify we shall write

\begin{equation}
E_n(x) = \left( (-1)^{n}\LintD(x, p_2,\ldots, p_{n}) \right)
\end{equation}

for $n\in \N$. Doing this gives us a sequence $(E_n(x))_{n\in \N}$ removes the alternating factor $(-1)^n$, allowing us to focus on the alternating sign over $m$.

For example

\[
E_3(x) = \log \frac{(p_1+x)(p_2+x)}{(p_1+p_2+x)(x)}.
\]

Note that all of the even subsets will now appear on the top of the fraction and the odd subsets will appear on the bottom. 

For the first case with $n=2$ we have
\begin{align}
\begin{split}
\frac{\partial \Lint}{\partial x} (x, p_2) &= E_2(x) \\
&= \log \frac{x+p_2}{x}
\end{split}
\end{align}
which is clearly greater than 0 for all $x \in \R^+$. The successive derivatives of $E_2(x)$ will continue to alternate in sign for $x\in \R^+$ using the standard power rule.

As we also know that $\Lint(x, p_2) = L(x, p_2) > 0$, the result holds for $n=2$. We now suppose that the statement is true for $n-1$.

We notice that

\begin{equation}
E_{n}(x) = E_{n-1}(x) - E_{n-1}(x+p_{n})
\end{equation}

Hence

\begin{align}
\label{EQN_difference_of_derivatives}
\begin{split}
& (-1)^n \frac{\partial^m \Lint}{\partial x^m} ( x, p_2, \ldots, p_n) \\
&=  \frac{\partial^{m-1}}{\partial x^{m-1}} E_n(x, p_2,\ldots, p_n) \\
&= \frac{\partial^{m-1}}{\partial x^{m-1}} E_{n-1}(x) - \frac{\partial^{m-1}}{\partial x^{m-1}} E_{n-1} (x + p_{n})
\end{split}
\end{align}

However by assumption we have that

\[
(-1)^{m-2}\frac{\partial^{m-2}}{\partial x^{m-2}} E_{n-1}(x) > 0
\]

Hence as the $m-2$-th partial derivative of $E_{n-1}$ has a given sign, we have that the difference between the terms of equation \eqref{EQN_difference_of_derivatives} has the opposite sign. That is,

\begin{equation}
(-1)^{m-1}\frac{\partial^{m-1}}{\partial x^{m-1}} E_n(x, p_2,\ldots, p_n) > 0
\end{equation}

Now, using lemma \ref{InteriorLossAt0} characterizing the interior loss at 0, and using that $E_n$ is strictly positive (negative) for all $x\in \R^+$, the sign of $\Lint$ will be strictly negative (positive) for $x\in \R^+$. Hence we have

\begin{equation}
(-1)^n (-1)^m \frac{\partial^{m} \Lint}{\partial x^{m}}(x, p_2\ldots, p_n) > 0.
\end{equation}

This completes the inductive argument.

\end{proof}

\subsection*{Proof of corollary \ref{InteriorMagnitude}}

\begin{proof}
We saw in lemma \ref{InteriorLossAt0} that it is sensible to extend $\Lint$ to $\R^+\cup \{0\}$ with $\Lint(p_1,\ldots, p_n) = 0$ when any $p_i = 0$. Moreover, as $\Lint$ is continuous as a function of $\tau$, varies strictly monotonically by lemma \ref{AlternatingDerivatives}, and is bounded at infinity by lemma \ref{LintAtInfinity}, we must have that $|\Lint(p_1,\ldots, p_{n-1}, \tau)| \in [0,|\Lint(p_1,\ldots, p_{n-1})|)$.
\end{proof}

\subsection*{Proof of theorem \ref{THM_yeung_correspondence}}
We first state a small lemma which is a standard property of entropy. We will make use of it to demonstrate that our measure is consistent with Yeung's $I$-measure.

\begin{lemma}
\label{LEMMA_partition_law}
Let $P_1,\ldots, P_k$ be disjoint subsets forming a partition of $\Omega$ consisting of individual outcomes $\omega$ of probability $p_\omega$. Then

\begin{equation}
L\left( \sum_{\omega \in P_1} p_\omega, \ldots, \sum_{\omega \in P_k} p_\omega  \right) = L(\Omega) - \sum_{i = 1}^k L(P_i).
\end{equation}

In particular, the expression of the left-hand side is equal to the measure of the subset $\Delta(\Omega) \setminus \left( \bigcup_{i=1}^k B(P_k) \right).$ 
\end{lemma}
\begin{proof}
We first demonstrate the simple identity
\begin{equation}
L(p_1+p_2,p_3,\ldots, p_n) = L(p_1,p_2,\ldots, p_n) - L(p_1, p_2).
\end{equation}
Let $\Omega = \{\omega_1,\ldots, \omega_N\}$. Then let $X$ be the random variable with partition $\{\{\omega_1,\omega_2\}, \{\omega_3\},\ldots, \{\omega_N\}\}$. By definition we have

\enlargethispage{-3cm} 

\begin{multline}
L(p_1, p_2) = H(\Omega)-H(X) \\
= L(p_1,\ldots, p_n) - L(p_1+p_2,\ldots, p_n),
\end{multline}

giving the identity. The full result then follows by symmetry on the arguments of $L$ and an inductive argument, sequentially decomposing sums into pairs.
\end{proof}

This result essentially states that the total loss of a certain variable defined by the partition $\{P_1,\ldots, P_k\}$ can be computed by calculating the total loss of the entire outcome space and subtracting boundaries internal to parts $P_i$.

We now proceed with the proof of the theorem.

\begin{proof}
We will show that our definition of content agrees with i.) the entropy of individual variables and ii.) the mutual information between two variables. The case for $n$ variables follows inductively.

We will now show that for a variable $X$ with an event space with associated probabilities $p_1,\ldots, p_n$, that $H(X) = L(p_1,\ldots, p_n) = \Lint(\content(X))$, the measure of the content in $X$ (see equation \eqref{EqnEntropyEqualToTotalLoss}).

Inside of a possibly more refined partition given by outcomes in $\Omega$, we can compute the entropy of $X$ by treating it as a partition $P_1,\ldots, P_k$ of the entire outcome space. In this case it is equivalent to the expression in lemma \ref{LEMMA_partition_law}. As mentioned after the lemma, this corresponds to the measure of the set

\begin{equation}
\Delta(\Omega) \, \setminus \, \left( \, \bigcup_{i=1}^k \, \left\{\Bint(S): S \subseteq P_i\right\} \right) = \content(X).
\end{equation}
It can be seen that this is equivalent to the construction of $\content(X)$ in definition \ref{DEFINITION_content}, as the only elements remaining in $\Delta(\Omega)$ must contain outcomes spanning across partitions. This completes i.).

The mutual information between two variables $X,Y$ is given by
\begin{equation}
\label{Eqn_standard_mutual_information}
I(X;Y) = H(X) + H(Y) - H(X,Y)
\end{equation}
We have seen that $H(V) = \Lint(\content(V))$ for a random variable $V$ inside of a refined space $\Omega$. Given two partitions $P$ and $Q$ corresponding to $X$ and $Y$ respectively, the collection generated by their intersections, $P_i\cap P_j$, is also a partition of $\Omega$, corresponding to the joint random variable $(X,Y)$. This is a refinement of the partitions of $X$ and $Y$.

In particular we have that $b\in \content(X)$ implies $b\in \content(X,Y)$. Constructing a formal sum of elements $b\in \content(X,Y)$, we can extend the measure $\Lint$ onto this formal sum to obtain

\begin{equation}
I(X;Y) = \Lint( \content(X) + \content(Y) - \content(X,Y) ) = \Lint(I)
\end{equation}

Where the formal sum $I = \content(X) + \content(Y) - \content(X,Y)$ will reflect the mutual information. We see that an atom $b\in \content(X,Y)$ does not appear in the formal sum $I$ unless $b \in \content(X) \cap \content(Y)$, in which case it appears with coefficient $1$. As all terms in the formal sum have coefficient $1$ or $0$, this formal sum also corresponds to the set of atoms in $\content(X)\cap \content(Y)$. Hence

\begin{equation}
I(X;Y) = \Lint(\content(X) \cap \content(Y)).
\end{equation}

That is, our logarithmic decomposition is consistent with standard Shannon mutual information and, by extension, all higher co-informations. It is hence a refinement of the $I$-measure of Yeung \cite{yeung1991new}.
\end{proof}

\subsection*{Proof of theorem \ref{THM_gacs_korner}}
\begin{proof}
The common information variable $Z$ is unique up to isomorphism, so it suffices to demonstrate that this variable $Z$ has its content $\content(Z) \subseteq \bigcap_i \content(X_i)$.

Given an outcome $\omega \in \Omega$, let $\omega$ be contained in the event $X_i(\omega)$ in $X_i$. That is, $\omega$ is contained in one of the parts $X_i(\omega)$ in the partition of $X_i$. By virtue of the definition of the common information, we must have

\begin{equation}
\label{EQN_function_requirements}
f_i(X_i(\omega)) = f_j(X_j(\omega)) \text{ for all $i,j\in\{1,\ldots, n\}$}.
\end{equation}

We will now show the result in two steps. Firstly we show that the common information variable induces a content in $\Delta(\Omega)$. Then we show that this is contained in the intersection $C$.

Viewing the random variables as partitions of $\Omega$ and using the ordering $A \leq B$ if $A$ is coarser than $B$, we obtain a lattice. Using the restriction in equation \eqref{EQN_function_requirements}, we can see that to compute the partition of $Z$ we must take the meet $X_1\land \cdots \land X_r$ of all variable partitions $X_i$ in the lattice. In particular, the partition of $Z$ has the property that $Z \leq \Omega$, and hence $\content(Z) \subseteq \Delta(\Omega)$, that is, we have the atoms needed to describe $Z$ in $\Delta(\Omega)$. Note that $\content(Z)$ might be empty, in which case it corresponds to the trivial random variable.

To show that $\content(Z)$ is contained in the intersection $C = \bigcap_{i} \content(X_i)$, let $\Bint(S) \in \content(Z)$. By definition, $S$ crosses a boundary in $Z$. As $Z$ is the finest partition which is coarser than $X_1,\ldots, X_r$, $S$ must cross a boundary in all $X_i$. That is, $\Bint(S) \in \bigcap_i \content(X_i)$. Hence $\content(Z) \subseteq C$.

Note that as the partition of $Z$ is unique, the content is also necessarily unique, giving the result.
\end{proof}

\end{document}